\documentclass[11pt,a4paper]{article}
\usepackage[left=2.8cm,right=2.8cm,bottom=4cm,top=4cm,bindingoffset=0.5cm]{geometry}
\usepackage[applemac]{inputenc}
\usepackage[english]{babel}
\usepackage{amsmath}
\usepackage{amssymb}
\usepackage{amsthm}
\usepackage{fancyhdr}
\usepackage{graphicx}
\usepackage{curves}
\usepackage{subfigure}
\usepackage{boxedminipage}
\usepackage{cite}
\usepackage{lastpage}
\usepackage{textcase}
\usepackage{authblk}
\theoremstyle{definition} \newtheorem{definition}{Definition}[section]
\theoremstyle{plain} \newtheorem{theorem}[definition]{Theorem}
\theoremstyle{plain} 
\theoremstyle{plain} \newtheorem{lemma}[definition]{Lemma}
\theoremstyle{plain} \newtheorem{corollary}[definition]{Corollary}
\theoremstyle{plain} \newtheorem{remark}[definition]{Remark}
\theoremstyle{definition} 
\numberwithin{equation}{section}
\fancypagestyle{plain}{%
  \fancyhf{}%
}

\newcommand{\N}{\mathbb{N}}
\newcommand{\Z}{\mathbb{Z}}

\newcommand{\R}{\mathbb{R}}

\newcommand{\E}{\mathcal{E}}

\newcommand{\Num}{\mathcal{N}}

\newcommand{\bra}{\langle}
\newcommand{\ket}{\rangle}
\newcommand{\A}{\mathcal{A}}
\newcommand{\M}{\mathcal{M}}
\newcommand{\F}{\mathcal{F}}

\newcommand{\supp}{\text{supp}}

\newcommand{\Oh}{\mathcal{O}}
\renewcommand{\epsilon}{\varepsilon}
\renewcommand{\phi}{\varphi}
\newcommand*{\mychi}{\raisebox{0.35ex}{\mbox{\large $\chi$}}}

\begin{document}

\title{The Ground State Energy of a Dilute Bose Gas \\ in Dimension $n>3$}
\author{Anders Aaen}
\affil{Department of Mathematics, Aarhus University\\ Ny Munkegade 118, 8000 Aarhus C, Denmark}
\date{\today}
\maketitle

\begin{abstract}
We consider a Bose gas in spatial dimension $n>3$ with a repulsive, radially symmetric two-body potential $V$. In the limit of low density $\rho$, the ground state energy per particle in the thermodynamic limit is shown to be $(n-2)|\mathbb S^{n-1}|a^{n-2}\rho$, where $|\mathbb S^{n-1}|$ denotes the surface measure of the unit sphere in $\R^n$ and $a$ is the scattering length of $V$. Furthermore, for smooth and compactly supported two-body potentials, we derive upper bounds to the ground state energy with a correction term $(1+C\gamma)8\pi^4a^6\rho^2|\ln(a^4\rho)|$ in dimension $n=4$, where $\gamma:=\int V(x)|x|^{-2}\, dx$, and a correction term which is $\Oh(\rho^2)$ in higher dimensions.
\end{abstract}


\section{Introduction}

The experimental realization of Bose-Einstein Condensation in 1995 \cite{bec1} has inspired renewed interest in a rigorous understanding of the interacting Bose gas, and in particular the ground state energy. The typical model for the energy of $N$ bosons enclosed in a box $\Lambda=\Lambda_L:=(-L/2,L/2)^n$, is the Hamiltonian 
\begin{equation} \label{model}
H_{N,L}=\sum_{i=1}^N-\Delta_i+\sum_{1\leq j<k\leq N} V(x_j-x_k)
\end{equation}
on $L^2_{\text{sym}}(\Lambda^N)$ (the set of totally symmetric $L^2$-functions on $\Lambda^N$). Here units are chosen such that $\hbar^2/2m=1$, where $m$ is the mass of a particle. We will always assume that the two-body potential $V$ is a nonnegative and radially symmetric function on $\R^n$. Let
$$
E_0(N,L):=\inf\sigma(H_{N,L})=\inf\{\bra \Psi, H_{N,L}\Psi \ket : \|\Psi\|=1\}
$$
denote the ground state energy of the Bose gas, and let
\begin{equation} \label{sdsdmdsnsdbsdh}
e_0(\rho):= \lim_{N\to \infty}\frac{E_0(N,(N/\rho)^{1/n})}{N}
\end{equation}
denote the ground state energy per particle in the thermodynamic limit at density $\rho>0$. The latter is independent of whatever boundary conditions imposed on $\Lambda$. We let $a$ denote the scattering length of $V$ (see section \ref{sec:leading}) and note that $Y:=a^n\rho$ is a dimensionless quantity.

In dimension $n=3$, the asymptotic behavior of $e_0(\rho)$ in the limit of low density was studied by Bogoliubov \cite{Bo}, Lee-Yang \cite{leeyang} and Lee-Huang-Yang \cite{leehuangyang} in the 1940-50's. In particular, the latter applied the pseudopotential method to derive the expansion
$$
e_0(\rho)=4\pi a \rho\bigg(1+\frac{128}{15\sqrt{\pi}}Y^{1/2}+ o\big(Y^{1/2}\big) \bigg) \quad \text{as $Y\to 0$},
$$
now known as the Lee-Huang-Yang formula (LHY). To give a mathematical proof of LHY is still an open problem, except in a special case of $\rho$ in a so-called simultaneously weak coupling and high density regime, and for a rather narrow class of potentials \cite {giulianiseiringer}. Even to prove the leading order term in LHY turned out to be a hard problem:  A variational calculation carried out by Dyson in 1957 \cite{dyson} showed the upper bound $e_0(\rho)\leq 4\pi a\rho(1+CY^{1/3})$, for hard-core interactions. This has later been generalized to general nonnegative, radially symmetric potentials \cite{liebseiringeryngvason2000}. However, no proof of a matching, leading order lower bound was available until 1998, where Lieb-Yngvason managed to show that $e_0(\rho)\geq 4\pi a\rho(1-CY^{1/17})$. Their approach was improved in \cite{lee-yin} to yield $e_0(\rho)\geq 4\pi a\rho (1-C\rho^{1/3}|\ln \rho|^3)$. At the present time, no lower bound has captured even the correct order in the expansion parameter $Y$ in LHY. For the upper bound there has been success though: In \cite{esy08} a trial state of the form
\begin{equation} \label{trial_esy}
\Psi=\exp\bigg(\frac{1}{2}\sum_{p\neq 0}c_pa_p^+a_{-p}^+ +\sqrt{N_0}a_0^+ \bigg) |0\ket
\end{equation}
was used to derive an upper bound
$$
e_0(\rho)\leq 4\pi a\rho \bigg(1+ \frac{128}{15\sqrt{\pi}}(1+C\lambda)Y^{1/2}\bigg) +\tilde  C\rho^2|\ln \rho| ,
$$
for a coupled two-body potential $V=\lambda \tilde V$. While the correction term has the correct order in $Y$, the constant is only correct in the limit of weak coupling, $\lambda\to 0$. The (Fock) trial state \eqref{trial_esy} is inspired by the Bogoliubov approximation, and the crucial feature is that particles of nonzero momenta appear only in \emph{pairs} of opposite momenta. Similar states have previously been considered by Girardeau-Arnowitt \cite{G-A} and Solovej \cite{Solovej} in the context of Bose gases. In a paper from 2009 \cite{YY} Yau-Yin introduced a new trial state, extending the properties of \eqref{trial_esy}. More precisely, they include pairs with total momentum of order $\rho^{1/2}$ (however their trial state has a fixed number of particles in contrast to \eqref{trial_esy}). This turns out to lower the energy significantly and their result is an upper bound consistent with LHY. We note however that the calculation with the Yau-Yin trial state is somewhat more involved than the computation with \eqref{trial_esy}.

The model \eqref{model} has also been studied in other dimensions. The case $n=1$ (with a delta-function potential) was already considered back in 1963 by Lieb-Liniger \cite{lieblinger} and turned out to be exactly solvable. In two dimensions, the leading order term was, to our knowledge, first identified by Schick \cite{schick} in 1971 to be $4\pi \rho |\ln(a^2\rho)|^{-1}$. This was rigorously proven to be correct by Lieb-Yngvason in 2001 \cite{liebyngvason2001}. To our knowledge there are yet no rigorous results on the $2$-dimensional correction term (in fact, it seems that there is not even consensus about what this term should be: compare e.g. \cite{schick}, \cite{yang} and \cite{mc}). In \cite{yang} Yang reexamined the pseudopotential method in dimension two, four and five. In the latter he found the method inconclusive, while in four dimensions he derived the expansion
\begin{equation} \label{yang4d}
e_0(\rho)=4\pi^2a^2\rho \big[1+2\pi^2Y |\ln Y| +o\big(Y\ln Y \big) \big]  \quad \text{as $Y\to 0$}.
\end{equation}
We remark that in Yangs paper the correction $2\pi^2Y |\ln Y|$ appears to be $4\pi^2Y |\ln Y|$, due to a minor miscalculation.

In this paper we test some of the rigorous $3$-dimensional calculations in higher dimensions. We follow the proofs of Dyson and Lieb-Yngvason to obtain the $n$-dimensional upper- and lower bounds (Theorem \ref{thm:dyson_upper} and Theorem \ref{thm:lower_bound}),
\begin{equation} \label{sdsdsdmsdnsdhsdysdu}
1-CY^{\alpha}\leq \frac{e_0(\rho)}{s_na^{n-2}\rho} \leq 1+CY^{\beta},
\end{equation}
where $s_n:=(n-2)|\mathbb{S}^{n-1}|$, $|\mathbb{S}^{n-1}|$ denotes the surface measure of the unit sphere in $\R^n$ and where
$$
\alpha=\frac{n-2}{n(n+2)+2} \quad \text{and} \quad \beta =\frac{n-2}{n}.
$$
Secondly, we employ the trial state \eqref{trial_esy} to improve the upper bounds. In dimension $n=4$ we show that (Theorem \ref{thm:ESYndim})
$$
e_0(\rho)\leq 4\pi^2a^2\rho\big[ 1+2\pi^2(1+C\gamma)Y|\ln Y| \big] +\Oh(\rho^2), 
$$
where $\gamma:=\int V(x)|x|^{-2}\, dx$, consistent with \eqref{yang4d} in the limit $\gamma \to 0$. In dimension $n\geq 5$ the calculation yields the upper bound (Theorem \ref{thm:ESYndim})
$$
e_0(\rho)\leq s_na^{n-2}\rho + \Oh(\rho^2) .
$$
The second order asymptotics of $e_0(\rho)$ becomes more subtle in dimension $n>3$. The correction to the energy is given in terms of certain integrals, which, in three dimensions, are exactly computable in the limit $\rho\to 0$, in a straight-forward manner. This is not the case in higher dimensions, and a more careful analysis has to be carried out. In dimension $n\geq 5$ we have not been able to identify the expansion parameter $Y$ in the correction term, nor an explicit coefficient. 

Finally, since \eqref{trial_esy} is a Fock state, we need the fact that the canonical ground state energy defined in \eqref{sdsdmdsnsdbsdh} can be recovered from the grand-canonical setting. Although this is a well-known result, we did not come across a good reference for it, and hence we have included a proof in Appendix \ref{app-ensembles}.

\section{The Leading Order Term} \label{sec:leading}

In this section we prove the upper and lower bounds in \eqref{sdsdsdmsdnsdhsdysdu}. We will assume that $V$ is a nonnegative, radial and measurable function on $\R^n$, where $n\geq 3$. The scattering length of $V$ is denoted by $a$ and may be defined via the variational problem (see e.g. \cite{LSSY}, \cite{yinphd})
\begin{equation} \label{def:scattt}
s_n a^{n-2}:=\inf_{u}\int_{\R^n}|\nabla u|^2+\frac{1}{2}Vu^2 ,
\end{equation}
where the infimum is taken over all nonnegative, radially symmetric functions $u\in H^1_{\text{loc}}(\R^n)$ satisfying $u(r)\to 1$ as $r\to \infty$. Notice that such functions are automatically continuous away from the origin. Also, it is easy to see that we may restrict attention to radially increasing functions. Moreover, we remark that $a$ is finite if and only if $V$ is integrable at infinity. In many cases the infimum in \eqref{def:scattt} is a unique minimum, and the minimizer $u$ satisfies the zero-energy scattering equation
\begin{equation} \label{sdsdsdnsdbsdhsdg}
-\Delta u +\frac{1}{2}Vu=0
\end{equation}
in the sense of distributions on $\R^n$. The existence of a scattering solution for a nonnegative, radially symmetric and \emph{compactly supported} potential is established in \cite{liebyngvason2001}. We note briefly some properties of the scattering solution $u$, referring to \cite{liebyngvason2001}, \cite{LSSY} for details:
\begin{itemize}
\item[(i)] For large $r$, $u(r)\approx 1- (a/r)^{n-2}$, or more precisely
\begin{equation} \label{sdsdsdmsdksdj}
\lim_{r\to \infty}\frac{1-u(r)}{(a/r)^{n-2}}=1.
\end{equation}
In fact
\begin{equation} \label{sssdsbdgsdfsdtsdrsy}
u(r)\geq 1- (a/r)^{n-2},
\end{equation}
with equality for $r>R_0$ if $\supp(V)\subset B(0,R_0)$.
\item[(ii)] Monotonicity: If $V\leq \tilde V$, then $a\leq \tilde a$, while $u \geq \tilde u$. 
\item[(iii)] Regularity imposed on $V$ is inherited by $u$. For instance, one may apply elliptic regularity and Sobolev imbedding's to show that if $V$ is smooth, so is $u$.
\item[(iv)]  For $V\in L^1(\R^n)$, it follows from \eqref{sdsdsdnsdbsdhsdg} that $u$ can be represented as
\begin{equation} \label{sdsddsmsdksdl}
1-u(x)=\frac{1}{2}\Gamma (Vu)(x):=\frac{1}{2s_n}\int_{\R^n}\frac{V(y)u(y)}{|x-y|^{n-2}}\, dy .
\end{equation}
By \eqref{sdsdsdmsdksdj} and the dominated convergence theorem it then follows that 
\begin{equation} \label{sdsdsmsdnsdydudtdy}
2s_na^{n-2}= \int_{\R^n}V(x)u(x)\, dx.
\end{equation}
\end{itemize}

The main result of this section is the following, which is an immediate consequence of Theorem \ref{thm:dyson_upper} and Corollary \ref{thm:lower_bound_cor} below.

\begin{theorem}
Let $n\geq 3$ and suppose that $V$ is nonnegative, radially symmetric, measurable and decays faster than $r^{-\nu}$ at infinity, where $\nu=(6n-2)/5$. Suppose furthermore that $V$ admits a scattering solution. Then
\begin{equation} \label{sdsdsdhsdbsdfsdgsdfsdg}
\lim_{\rho \to 0} \frac{e_0(\rho)}{s_na^{n-2}\rho} =1.
\end{equation}
\end{theorem}

\subsection{The Upper Bound}

We have the following dimensional generalization of  \cite{dyson}, \cite{LSSY}.

\begin{theorem} \label{thm:dyson_upper}
Let $n\geq 3$ and suppose that $V$ is nonnegative, radially symmetric and measurable.
\begin{itemize}
\item[(i)] Without further assumptions,
$$
\limsup_{\rho \to 0}\frac{e_0(\rho)}{s_na^{n-2}\rho}\leq 1.
$$
\item[(ii)]  There exist $C,\delta>0$ independent of $V$ such that, if $V$ admits a scattering solution, then 
$$
e_0(\rho)\leq s_na^{n-2}\rho\big[1+CY^{1-2/n} \big],
$$
whenever $Y\leq \delta$.
\end{itemize}
\end{theorem}

\begin{proof}
We employ the periodic trial state of Dyson \cite{dyson}. This state is not symmetric, but since the ground state of $H_{N,L}$ on the full space $L^2(\Lambda^N)$ \emph{is} symmetric \cite{LSSY}, we obtain an upper bound to $e_0(\rho)$.  Suppose that $u\in H^1_{\text{loc}}(\R^n)$ is nonnegative, radially symmetric, increasing and moreover that $u(r)\to 1$ as $r\to \infty$. The trial state is then defined by
$$
\Psi:=F_2\cdot F_3\cdots F_N,
$$
where
$$
F_i:=\min_{1\leq j<i}\bigg[\min_{m\in \Z}f(x_i-x_j-mL) \bigg]
$$
and 
$$
f(r):=\left\{\begin{array}{ll} \frac{u(r)}{u(b)} & 0\leq r\leq b \\ 1 & r>b
 \end{array} \right. ,
$$
for some (large) $b>0$ to be chosen. Following the calculation in \cite{LSSY} we obtain
\begin{equation} \label{sdsdjsdnsdbsdgsdt}
e_0(\rho)\leq \frac{J\rho+\frac{2}{3}(K\rho)^2}{(1-I\rho)^2},
\end{equation}
where 
$$
I:=\int (1-f(x)^2)\, dx, \quad K:=\int f(x)|\nabla f(x)|\, dx
$$
and
$$
J:=\int |\nabla f(x)|^2+\frac{1}{2}V(x)f(x)^2\, dx.
$$
It follows that
$$
\limsup_{\rho\to 0}e_0(\rho)\rho^{-1}\leq J \leq \frac{1}{u(b)^2}\int |\nabla u(x)|^2+\frac{1}{2}V(x)u(x)^2 \, dx,
$$
where we have used 
$$
f(r)\leq \frac{u(r)}{u(b)} \quad \text{and} \quad  f'(r)\leq  \frac{u'(r)}{u(b)}
$$
in the latter inequality. In the limit $b\to \infty$ we get
$$
\limsup_{\rho\to 0}e_0(\rho)\rho^{-1}\leq \int |\nabla u(x)|^2+\frac{1}{2}V(x)u(x)^2 \, dx,
$$
and minimizing over $u$ yields (i), by definition of the scattering length.

In case $V$ admits a scattering solution, we apply the above construction with $u$ being this particular function. The bound \eqref{sssdsbdgsdfsdtsdrsy} then allows us to estimate more explicitly. Indeed, we have $f(r)\geq [1-(a/r)^{n-2}]_+$, and hence
$$
I\leq |\mathbb{S}^{n-1}| \bigg(\int_0^a r^{n-1}\, dr +\int_a^b2a^{n-2}r\, dr \bigg) \leq |\mathbb{S}^{n-1}|a^{n-2}b^2 .
$$
Next,
$$
J\leq \frac{s_n a^{n-2}}{u(b)^2} \leq  \frac{s_n a^{n-2}}{(1-(a/b)^{n-2})^2},
$$
provided $b>a$. Finally, using $f(r)\leq 1$ and an integration by parts yields
$$
K\leq |\mathbb{S}^{n-1}|\int_0^bf'(r)r^{n-1}\, dr \leq |\mathbb{S}^{n-1}| \left(b^{n-1}-(n-1)\int_0^bf(r)r^{n-2}\, dr \right) .
$$
However,
\begin{eqnarray*}
\int_0^bf(r)r^{n-2}\, dr &\geq& \int_a^b\big[1-(a/r)^{n-2} \big]r^{n-2}\, dr\\
&=& \frac{b^{n-1}}{n-1}-a^{n-2}b+\frac{n-2}{n-1}a^{n-1} \geq \frac{b^{n-1}}{n-1}-a^{n-2}b,
\end{eqnarray*}
and hence $K\leq |\mathbb{S}^{n-1}|(n-1)a^{n-2}b$. Now, by choosing $b:=(|\mathbb{S}^{n-1}|\rho)^{-1/n}$, we have
$$
(a/b)^{n-2}=|\mathbb{S}^{n-1}|a^{n-2}b^2\rho=\tilde Y^{\beta},
$$
where $\tilde Y:=|\mathbb{S}^{n-1}| Y$ and $\beta:=(n-2)/n$. Note that in particular $b>a$ if $\tilde Y<1$. In total we have
$$
e_0(\rho)\leq s_na^{n-2}\rho \bigg[\frac{1}{(1-\tilde Y^{\beta})^4}+\frac{C Y^{\beta}}{(1-\tilde Y^{\beta})^2} \bigg] \leq s_na^{n-2}\rho\big(1+\tilde C Y^{\beta} \big),
$$
provided $\tilde Y$ is bounded away from $1$.
\end{proof}

\subsection{The Lower Bound}

In this section we prove an $n$-dimensional lower bound by following the steps in \cite{liebyngvason98}. The assumption of compact support in Theorem \ref{thm:lower_bound} below is relaxed in Corollary \ref{thm:lower_bound_cor}.

\begin{theorem} \label{thm:lower_bound} 
Let $n\geq 3$ and suppose that $V$ is nonnegative, radially symmetric, measurable and compactly supported with, say, $\text{supp}(V)\subset B(0,R_0)$. There exist $C,\delta>0$ independent of $V$ such that
$$
e_0(\rho)\geq s_na^{n-2}\rho \big(1-C Y^{\alpha} \big),
$$
where
\begin{equation} \label{sdsdsdsdsdmsdnsdhsdjsdysdu}
\alpha := \frac{n-2}{n(n+2)+2},
\end{equation}
provided 
\begin{equation} \label{qqmqmqkqmqj}
Y\leq \min \big\{\delta, (a/R_0)^{\frac{n-2}{5\alpha}} \big\}.
\end{equation}
\end{theorem}

In order to prove Theorem \ref{thm:lower_bound} we consider $H=H_{N,L}$ with Neumann boundary conditions on $\Lambda$. The first step is to obtain an $n$ - dimensional version of Dyson's lemma. In what follows we set
$$
a_n:=(n-2)a^{n-2}.
$$
\begin{lemma}[Dyson's Lemma] \label{lem:dyson}
Suppose that $U$ is a measurable, nonnegative and radially symmetric function on $\R^n$, which satisfies
$$
U(r)=0, \quad \text{for $r\leq R_0$}, \quad \text{and} \quad \int_0^{\infty}U(r)r^{n-1}\, dr\leq 1.
$$
Let $B\subseteq \R^n$ be open and star shaped w.r.t. the origin. Then
$$
\int_{B}  |\nabla \phi(x)|^2+\frac{1}{2}V(x)|\phi(x)|^2\, dx \geq a_n \int_{B}U(x)|\phi(x)|^2\, dx, 
$$
for each $\phi\in H^1(B)$.
\end{lemma}

\begin{proof}
For any $\omega\in \mathbb{S}^{n-1}$ we let
$$
R(\omega)=\sup\{r\geq 0: s\omega \in B, \, \text{for each $0\leq s\leq r$}\}
$$
denote the (possibly infinite) distance from the origin to the boundary of $B$ in the direction of $\omega$. Since $B$ is open and star shaped w.r.t. the origin, it follows that, for any $r\geq 0$, $r\omega \in B$ if and only if $r<R(\omega)$. By passing into polar coordinates, we then see that it suffices to show that, for each fixed $\omega\in \mathbb{S}^{n-1}$,
\begin{equation} \label{eq:sdfjsdfisdfo}
\int_0^{R(\omega)}\big( |f'(r)|^2 + \frac{1}{2}V(r)|f(r)|^2 \big)r^{n-1}\, dr \geq a_n \int_0^{R(\omega)}U(r)|f(r)|^2 r^{n-1}\, dr,
\end{equation}
where $f(r):=\phi(r\omega)$ with $|f'(r)|\leq |\nabla \phi(r\omega)|$. We may assume that $R(\omega)>R_0$, since otherwise the right hand side in \eqref{eq:sdfjsdfisdfo} vanishes, and we claim that
\begin{equation} \label{eq:sdksdosdlsdp}
\int_0^{R(\omega)}\big( |f'(r)|^2 + \frac{1}{2}V(r)|f(r)|^2 \big)r^{n-1}\, dr \geq a_n|f(R)|^2, 
\end{equation}
for each $R_0<R<R(\omega)$. Indeed, if $f(R)\neq 0$, then the function $u$ given by $u(x)=|f(|x|)/f(R)|$ for $|x|\leq R$ and $u(x)=1$ for $|x|>R$ is admissible in \eqref{def:scattt}, and since $V(r)=0$, for $r>R$, it  follows that
$$
s_na^{n-2}\leq \frac{|\mathbb{S}^{n-1}|}{|f(R)|^2}\int_0^{R(\omega)}\big( |f'(r)|^2 + \frac{1}{2}V(r)|f(r)|^2 \big)r^{n-1}\, dr .
$$
Now \eqref{eq:sdfjsdfisdfo} follows by multiplying both sides of \eqref{eq:sdksdosdlsdp} with $U(R)R^{n-1}$ and then integrating w.r.t. $R$.
\end{proof}

\begin{corollary} \label{cor:saasjasdjasd}
Suppose that $U$ satisfies the conditions of Lemma \ref{lem:dyson}, and define
$$
W:=\sum_{i=1}^N U\circ t_i, \quad t_i(x_1,\ldots,x_N):=\min_{j\neq i}|x_i-x_j| .
$$
Then $H\geq a_nW$.
\end{corollary}

\begin{proof}
Since $V$ is nonnegative and radial,
$$
\sum_{i=1}^NV(t_i(\vec x)) \leq \sum_{i=1}^N\sum_{j<i}V(x_i-x_j)+\sum_{i=1}^N\sum_{j>i}V(x_i-x_j) =2\sum_{i<j}V(x_i-x_j),
$$
for each $\vec x=(x_1,\ldots,x_N)$, and hence
\begin{equation} \label{sdsdsdnsdbsdgsdvsdf}
H\geq \sum_{i=1}^N\big(-\Delta_i+\frac{1}{2}V\circ t_i \big) .
\end{equation}
We focus on the first term $i=1$, and fix $x_2,\ldots,x_N\in \Lambda$. For $j\neq 1$ define
$$
B_j=\{x_1\in \Lambda : t_1(\vec x)=|x_1-x_j|\} .
$$
Fix an arbitrary $\psi\in H^1(\Lambda^N)$. By a change of variables $x_1\mapsto x_1+x_j$, and by noting that $(B_j-x_j)$ is star shaped w.r.t. the origin (indeed convex), we may apply Dyson's lemma to obtain
\begin{equation} \label{eq:sdjsdudshrl}
\int_{B_j} |\nabla_1\psi(\vec x)|^2+\frac{1}{2}V(t_i(\vec x))|\psi(\vec x)|^2\, d x_1 \geq a_n\int_{B_j} U(t_1(\vec x))|\psi(\vec x)|^2\, dx_1 ,
\end{equation}
for each $j\neq 1$. Moreover, since the $B_j$'s cover $\Lambda$ disjointly (a.e.), we conclude that \eqref{eq:sdjsdudshrl} holds with $B_j$ replaced by $\Lambda$. Then, by Fubini's theorem,
$$
\int_{\Lambda^N} |\nabla_1\psi(\vec x)|^2+\frac{1}{2}V(t_1(\vec x))|\psi(\vec x)|^2\, d\vec x \geq a_n\int_{\Lambda^N} U(t_1(\vec x))|\psi(\vec x)|^2 \, d\vec x.
$$
We get analogous contributions from $i=2,\ldots,N$ in \eqref{sdsdsdnsdbsdgsdvsdf}, and upon adding them, we obtain the result.
\end{proof}

We now combine Corollary \ref{cor:saasjasdjasd} with Temple's inequality \cite{RS4} in a perturbative approach. The parameters $R$ and $\epsilon$ appearing below will be chosen appropriately later on.
\begin{lemma} \label{sdsdmsdlsdksdodsp}
Let $0<\epsilon<1$ and $R_0<R<L/2$. Suppose that
$$
G(N,L):=\epsilon\pi^2L^{-2}-s_na^{n-2} L^{-n}N^2 > 0.
$$
Then
$$
E_0(N,L)\geq N(N-1)K(N,L),
$$
where
$$
K(N,L):=\frac{s_na^{n-2}}{L^n} (1-\epsilon)\bigg(1-\frac{2R}{L}\bigg)^n \bigg(1-v_n\frac{R^n}{L^n}\bigg)^{N-2}  \bigg(1-\frac{n(n-2)a^{n-2}N}{(R^n-R_0^n)G(N,L)} \bigg)
$$
Here $v_n$ denotes the measure of the unit ball in $\R^n$.
\end{lemma}

\begin{proof}
Suppose that $U$ and $W$ are as in Lemma \ref{lem:dyson} respectively Corollary \ref{cor:saasjasdjasd}. Together with the fact that $V$ is nonnegative, we then have a lower bound
$$
H=\epsilon H +(1-\epsilon)H \geq  -\epsilon \Delta +(1-\epsilon)a_n W =:\tilde H ,
$$
and consequently
\begin{equation} \label{sdsdsdmsdnsdhsdjsdg}
E_0(N,L)\geq \tilde E_0(N,L):=\inf \sigma(\tilde H).
\end{equation}
We estimate $\tilde E_0(N,L)$ by employing Temple's inequality in the ground state of $-\epsilon \Delta$ (with Neumann Boundary conditions), which is the constant function $\phi_0(x)\equiv |\Lambda|^{-N/2}$ with corresponding eigenvalue zero. Given any operator $A$ on $L^2(\Lambda^N)$ with domain containing $\phi_0$, we let $\bra A\ket=\bra \phi_0,A\phi_0\ket$. Temple's inequality and \eqref{sdsdsdmsdnsdhsdjsdg} yields
\begin{eqnarray*}
E_0(N,L)&\geq& \bra \tilde H\ket -\frac{\bra \tilde H^2\ket -\bra \tilde H \ket^2}{\tilde E_1-\bra \tilde H \ket} \\
&=& (1-\epsilon)a_n\bra W\ket - \frac{(1-\epsilon)^2a_n^2\big(\bra W^2\ket -\bra W\ket^2 \big)}{\tilde E_1-(1-\epsilon)a_n\bra W\ket},
\end{eqnarray*}
provided $\bra \tilde H\ket<\tilde E_1$, where $\tilde E_1$ is the second lowest eigenvalue of $\tilde H$. Note however that, since $W$ is nonnegative, we have $\tilde H \geq -\epsilon \Delta$, and hence $\tilde E_1 \geq \epsilon \pi^2/L^2$, which is the second lowest eigenvalue of $-\epsilon \Delta$. We now choose the function $U$ to be 
$$
U(r):=\left\{ \begin{array}{ll}  
n(R^n-R_0^n)^{-1} & \text{for $R_0<r<R$} \\
0 & \text{otherwise}
\end{array} \right. .
$$
By discarding the term $\bra W\ket^2$, replacing $(1-\epsilon)$ by $1$ in two appropriate places and employing the fact that
$$
\bra W^2\ket \leq n\cdot N (R^n-R_0^n)^{-1}\bra W\ket,
$$
we obtain
\begin{equation}\label{sdlsdsdpsdosdpsdi}
E_0(N,L)\geq (1-\epsilon)a_n \bra W\ket \bigg[1-\frac{n a_n N}{(R^n-R_0^n)\big(\epsilon\pi^2/L^2-a_n\bra W\ket \big)} \bigg],
\end{equation}
provided $a_n\bra W\ket <\epsilon\pi^2/L^2$. To estimate this further, we need upper and lower bounds on $\bra W\ket$, and we claim that
\begin{equation}\label{adsmasdjsaduasdyasdusadi}
\frac{|\mathbb{S}^{n-1}|}{L^{n}}N(N-1)\big(1-2R/L \big)^n \big(1-v_n R^n/L^n \big)^{N-2} \leq \bra W\ket \leq \frac{|\mathbb{S}^{n-1}|}{L^{n}} N(N-1) .
\end{equation}
This will conclude the proof of the lemma. For the upper bound in \eqref{adsmasdjsaduasdyasdusadi} we fix $x_1\in \Lambda$ and notice that
$$
 \big\{(x_2,\ldots,x_N)\in \Lambda^{N-1}: R_0<t_1(\vec x)<R\big\} \subseteq \bigcup_{j=2}^N F_j,
$$
where $\vec x=(x_1,\ldots,x_N)$ and $F_j=\Lambda^{N-1}$, except that the $j$'th factor is replaced by $B(x_1,R)\backslash B(x_1,R_0)$. It follows that
$$
\int_{\Lambda^{N-1}}U(t_1(\vec x))\, dx_2\ldots dx_N \leq \frac{n}{R^n-R_0^n} \sum_{j=2}^N |F_j| = |\mathbb{S}^{n-1}|(N-1)|\Lambda|^{N-2}.
$$
By integrating over $x_1\in \Lambda$ and then adding the identical contributions from the integrals of $U(t_2),\ldots,U(t_N)$, we arrive at the upper bound in \eqref{adsmasdjsaduasdyasdusadi}. To verify the lower bound, we let $\Lambda'\subseteq \Lambda$ denote the cube with same center as $\Lambda$ but with side length $L-2R$. Fix $x_1\in \Lambda'$ and notice that $B(x_1,R)\subseteq \Lambda$. We then have
\begin{equation} \label{eq:ssdhdsnsdvsdb}
\bigcup_{j=2}^N E_j \subseteq  \big\{(x_2,\ldots,x_N)\in \Lambda^{N-1}: R_0<t_1(\vec x)<R\big\},
\end{equation}
where
$$
E_j=(\Lambda\backslash B(x_1,R))^{N-1} 
$$
except again that the $j$'th factor is replaced by $B(x_1,R)\backslash B(x_1,R_0)$. Since the $E_j$'s are pairwise disjoint, \eqref{eq:ssdhdsnsdvsdb} implies that
$$
 \int_{\Lambda^{N-1}} U(t_1(\vec x))\, dx_2\ldots dx_N \geq \frac{n}{R^n-R_0^n} \sum_{j=2}^N |E_j| = |\mathbb{S}^{n-1}|(N-1)(|\Lambda|-v_nR^n)^{N-2},
$$
and integrating over $\Lambda \supset \Lambda'\ni x_1$, we obtain
$$
 \int_{\Lambda^{N}} U(t_1(\vec x))\, d\vec x \geq |\mathbb{S}^{n-1}|(N-1)(L-2R)^n (|\Lambda|-v_nR^n)^{N-2} .
$$
Again, by adding identical contributions from the integrals of  $U(t_2),\ldots,U(t_N)$, we have proved \eqref{adsmasdjsaduasdyasdusadi} and with it the lemma.
\end{proof}
 
Note that, for fixed $\rho>0$,
 $$
 G(\rho L^n,L)\leq \pi^2 L^{-2}-s_na^{n-2}\rho^2L^n <0,
 $$
for large $L$, so Lemma \ref{sdsdmsdlsdksdodsp} may not immediately be applied to get estimates in the thermodynamic limit.

\begin{lemma}\label{lem:super}
The mapping $N\mapsto E_0(N,L)$ is superadditive, i.e.,
$$
E_0(k+m,L)\geq E_0(k,L)+E_0(m,L), \quad \text{for all $k,m\in \N$}.
$$
\end{lemma}

\begin{proof}
Fix an arbitrary normalized $\psi\in H^1(\Lambda^{k+m})$. Since $V$ is nonnegative, it follows that
\begin{eqnarray}  \label{sadjasdkasdiasdo}
\bra \psi, H\psi\ket &\geq& \int_{\Lambda^{k+m}}  \sum_{i=1}^{k}|\nabla_i \psi|^2 + \sum_{1\leq i<j\leq k}V(x_i-x_j) |\psi|^2 \\
&+&  \int_{\Lambda^{k+m}} \sum_{i=k+1}^{k+m}|\nabla_i \psi|^2 + \sum_{k+1\leq i<j\leq k+m}V(x_i-x_j) |\psi|^2 . \nonumber
\end{eqnarray}
Then, by Fubini's theorem, 
$$
\int_{\Lambda^{k+m}} \sum_{i=1}^{k}|\nabla_i \psi|^2 + \sum_{1\leq i<j\leq k}V(x_i-x_j) |\psi|^2 \geq \int_{\Lambda^m}\left(E_0(k,L)\int_{\Lambda^k}|\psi|^2\right)=E_0(k,L),
$$
and similarly for the second term on the right-hand side in \eqref{sadjasdkasdiasdo} .
\end{proof}

\begin{lemma}
Suppose that $L/l\in \N$. Then 
\begin{equation} \label{cellbound}
E_0(N,L)\geq M \cdot \min \sum_{m=0}^Nc_mE_0(m,l),
\end{equation}
where $M:=(L/l)^n$ and where the minimum is over all tuples $(c_0,\ldots,c_N)$ of numbers $c_m\geq 0$ subject to the conditions
\begin{equation}\label{eq:constraints}
\sum_{m=0}^Nc_m=1 \quad \text{and} \quad \sum_{m=0}^Nmc_m=N/M .
\end{equation}
\end{lemma}

\begin{proof}
We partition $\Lambda$ into $M$ disjoint boxes $\Lambda_1,\ldots,\Lambda_M$, each of side length $l$. Correspondingly we have a partition $\{\Omega_{\beta}\}$ of $\Lambda^N$,
$$
\Omega_{\beta}:=\Lambda_{\beta_1}\times \ldots\times \Lambda_{\beta_N}, \quad \beta=(\beta_1,\ldots,\beta_N), \; 1\leq \beta_j\leq M ,
$$
and hence
\begin{equation} \label{sdsdsdmsdhsdjsdu}
\bra \psi, H\psi \ket = \sum_{\beta}\int_{\Omega_{\beta}} \sum_{i=1}^N|\nabla_i\psi|^2+\sum_{i<j}V(x_i-x_j)|\psi|^2 .
\end{equation}
Fix a $\beta$ as above. By Fubini's theorem, the integration regime $\Omega_{\beta}$ may be replaced by $\Lambda_1^{\alpha_1}\times\ldots \times \Lambda_M^{\alpha_M}$, for some multiindex $\alpha\in \N_0^M$ with length $|\alpha|=N$. For each $0\leq m\leq N$, we let $M\cdot c_m$ denote the number of components of $\alpha$ equal to $m$. By following the proof of Lemma \ref{lem:super}, we split the kinetic energy into appropriate terms, and discard interactions between particles in different boxes to obtain the lower bound
\begin{eqnarray*}
\int_{\Omega_{\beta}} \ldots &\geq & \bigg(\int_{\Omega_\beta}|\psi|^2 \bigg)\sum_{j=1}^ME_0(\alpha_j,l) \\
&=&  \bigg(\int_{\Omega_\beta}|\psi|^2 \bigg)  M\sum_{m=0}^Nc_mE_0(m,l) \\
&\geq& \bigg(\int_{\Omega_\beta}|\psi|^2 \bigg)  M \min \bigg(\sum_{m=0}^Nc_mE_0(m,l)\bigg) .
\end{eqnarray*}
Employing this estimate in \eqref{sdsdsdmsdhsdjsdu} yields the result.
\end{proof}

\begin{lemma} \label{asfasraseast}
Let  $\rho=N/L^n$. Suppose that $L/l\in \N$, $R_0<R<l/2$ and $ G(4\rho l^n,l)>0$. Then
$$
\frac{E_0(N,L)}{N}\geq (\rho l^n-1)K(4\rho l^n,l).
$$
\end{lemma}

\begin{proof}
Suppose that $c_m\geq 0$ satisfies \eqref{eq:constraints}. We split the sum in \eqref{cellbound} into two parts:
\begin{equation} \label{eq:sdjsdusdnsdhsdy}
\sum_mc_mE_0(m,l)=\sum_{m<p}c_mE_0(m,l) + \sum_{m\geq p} c_mE_0(m,l),
\end{equation}
for some $p\in \N$ to be chosen. Suppose for now that $G(p,l)>0$.  Since $G(N,L)$ and $K(N,L)$ are decreasing functions of $N$, Lemma \ref{sdsdmsdlsdksdodsp} implies that
\begin{equation}\label{eq:sasmasbasvascasv}
E_0(m,l)\geq  m(m-1)K(p,l), \quad \text{for $0\leq m\leq p$},
\end{equation}
and hence
$$
\sum_{m<p}c_mE_0(m,l)  \geq K(p,l) \sum_{m<p}c_mm(m-1) .
$$
Let $t:=\sum_{m<p}mc_m$. By the Cauchy-Schwarz inequality,
$$
t^2 \leq \left( \sum_{m<p}m^2c_m\right) \left(\sum_{m<p}c_m \right) \leq \sum_{m<p}m^2c_m,
$$
and it follows that
$$
\sum_{m<p}c_mm(m-1) \geq t(t-1) .
$$
Thus we have 
$$
\sum_{m<p}c_mE_0(m,l)  \geq  K(p,l)t(t-1) .
$$
We now employ the superadditivity of $m\mapsto E_0(m,l)$ (Lemma \ref{lem:super}) to obtain a lower bound on the second sum on the right hand side in \eqref{eq:sdjsdusdnsdhsdy}. For $m\geq p$ we write $m=\lfloor m/p\rfloor p+r$, where $\lfloor m/p\rfloor$ denotes the lower integer part of $m/p$ and $r\in \N_0$ is the remainder. Notice that $\lfloor m/p\rfloor\geq m/(2p)$ always. The superadditivity of $E_0(m,l)$ then yields
$$
E_0(m,l)\geq m/(2p) E_0(p,l) ,
$$
and it follows that
$$
\sum_{m\geq p}c_mE_0(m,l)\geq  \frac{E_0(p,l)}{2p}(k-t) \geq \frac{1}{2}(p-1)(k-t) K(p,l),
$$
where $k:=N/M=\rho l^n$. Altogether we have
$$
\sum_{m=0}^Nc_m E_0(m,l) \geq  K(p,l)\big[t(t-1)+\frac{1}{2}(p-1)(k-t) \big] .
$$
The choice $p=\lfloor 4k\rfloor$ implies that $x\mapsto \big(x(x-1)+\frac{1}{2}(p-1)(k-x) \big)$ is decreasing on $[0,k]$, which is where $t$ lies, and hence the minimum is taken at $x=k$. Thus we have that
$$
\frac{E_0(N,L)}{N}\geq \frac{1}{\rho l^n}\sum_m c_mE_0(m,l)\geq K(p,l)(k-1),
$$
as claimed.
\end{proof}

We can now finish the proof of Theorem \ref{thm:lower_bound}.

\begin{proof}[Proof of Theorem \ref{thm:lower_bound}]
Suppose that the conditions of Lemma \ref{asfasraseast} are satisfied. Recall that $Y=a^n\rho$. Then
\begin{eqnarray*}
\frac{E_0(N,L)}{N}&\geq& s_na^{n-2}\rho \big(1-\epsilon\big)\big(1-2nR/l\big)\bigg[1-Y^{-1} (a/l)^n \bigg] \\
&\times& \bigg[1-4v_nY (l/a)^n (R/l)^n \bigg] \bigg[1 - \frac{4n(n-2)l^nY}{ (R^n-R_0^n)(\epsilon \pi^2 (a/l)^2-16s_nY^2(l/a)^n)} \bigg] .
\end{eqnarray*}
We now make the ansatz
$$
\epsilon =Y^{\alpha}, \quad a/l=Y^{\beta}, \quad \frac{R^n-R_0^n}{l^n}=Y^{\gamma},
$$
for exponents $\alpha,\beta,\gamma>0$. In particular this implies that
$$
\bigg(\frac{R}{l}\bigg)^n= Y^{\gamma}+\bigg(\frac{R_0}{a}\bigg)^nY^{n\beta}\leq 2Y^{\gamma},
$$
provided 
\begin{equation} \label{sdsdsdksdnsdh}
Y\leq \bigg(\frac{a}{R_0} \bigg)^{n/(n\beta-\gamma)}.
\end{equation}
Thus we have
\begin{eqnarray*}
\frac{E_0(N,L)}{N}&\geq& s_na^{n-2}\rho \big(1-Y^{\alpha} \big)\big(1-C_1Y^{\gamma/n} \big)\big(1-Y^{n\beta-1} \big)\big(1-C_2Y^{1+\gamma-n\beta} \big) \\
&\times& \bigg(1-\frac{C_3Y^{1-\alpha-2\beta-\gamma}}{1-C_4Y^{2-\alpha-(n+2)\beta}} \bigg).
\end{eqnarray*}
In attempt to fit exponents we choose $\beta$ and $\gamma$ such that
$$
\gamma/n=\alpha= n\beta -1,
$$
which in particular implies that $1+\gamma-n\beta=2\alpha$. Now, the optimal choice of $\alpha$, such that
$$
1-\alpha-2\beta-\gamma \geq \alpha \quad \text{and} \quad 2-\alpha-(n+2)\beta>0,
$$
is given in \eqref{sdsdsdsdsdmsdnsdhsdjsdysdu}. With this choice the requirements of Lemma \ref{asfasraseast} are indeed satisfied if $Y$ is sufficiently small (depending only on the dimension) and if we take $L=kl$, for an integer $k\in \N$. Also \eqref{sdsdsdksdnsdh} is exactly the latter condition in \eqref{qqmqmqkqmqj}. By letting $k\to \infty$ we therefore conclude the proof. 
\end{proof}

\begin{corollary} \label{thm:lower_bound_cor} 
Suppose that $V$ is nonnegative, radial and measurable with a decay $V(r)\leq Cr^{-\nu}$, for large $r$, where $\nu>(6n-2)/5$. Suppose furthermore that $V$ admits a scattering solution. There exist a constant $C>0$ depending only on $n$ and a $\delta>0$ depending on $n,V$ such that
$$
e_0(\rho)\geq s_na^{n-2}\rho \big(1-C Y^{\alpha} \big),
$$
provided $Y\leq \delta$.
\end{corollary}

\begin{proof}
Let $R>0$ and define $V_R=V\mychi_{B(0,R)}$ with scattering length $a_R\leq a$. Since $V$ is nonnegative, replacing $V$ with $V_R$ cannot increase the energy. By Theorem \ref{thm:lower_bound} we then have
\begin{eqnarray*}
e_0(\rho) \geq s_na_R^{n-2}\rho \big(1-CY_R^{\alpha} \big)  \geq s_na_R^{n-2}\rho \big(1-CY^{\alpha} \big) ,
\end{eqnarray*}
provided $Y_R:=a_R^n\rho$ is sufficiently small and 
\begin{equation} \label{sdsdsdjsdnsdh}
Y_R\leq \bigg(\frac{a_R}{R}\bigg)^{\frac{n-2}{5\alpha}} .
\end{equation}
Denote the scattering solutions of $V$ and $V_R$ by $u$ respectively $u_R$. Then, by \eqref{sdsdsmsdnsdydudtdy},
\begin{eqnarray*}
a^{n-2}-a_R^{n-2}&=&\frac{1}{2s_n}\int V(x)u(x)-V_R(x)u_R(x) \, dx \\
&\leq& \frac{1}{2s_n} \int V(x)-V_R(x) \, dx = \frac{1}{2s_n} \int_{|x|\geq R}V(x)\, dx,
\end{eqnarray*}
where the inequality follows from the fact that $u\leq u_R\leq 1$. From the decay of $V$ we obtain
$$
a_R^{n-2}\geq a^{n-2}\bigg(1-\frac{K}{2(n-2)a^{n-2}R^{\epsilon}} \bigg),
$$
provided $R$ is sufficiently large. By choosing $R$ such that 
$$
\frac{K}{2(n-2)a^{n-2}R^{\epsilon}} =Y^{\alpha},
$$
it follows that $R$ is large,
$$
a_R^{n-2}\geq a^{n-2}\big(1 -Y^{\alpha}\big) ,
$$
and \eqref{sdsdsdjsdnsdh} is satisfied, if $Y$ is sufficiently small and $\nu>(6n-2)/5$.
\end{proof}

\section{A Second Order Upper Bound} \label{sec:second}

In this section we derive a second order upper bound to $e_0(\rho)$ by estimating the energy in the state \eqref{trial_esy}. The calculation is inspired by \cite{esy08}.

\begin{theorem} \label{thm:ESYndim}
Let $n\geq 3$ and suppose that $V\in C_0^{\infty}(\R^n)$ is nonnegative and radially symmetric with $V(0)>0$. Then
\begin{align*}
e_0(\rho)&\leq 4\pi a\rho \bigg(1+[1+C\gamma]\frac{128}{15\sqrt{\pi}}Y^{1/2} \bigg) +\Oh(\rho^2|\ln \rho|) \quad (n=3) \\
e_0(\rho)&\leq 4\pi^2 a^2\rho \bigg(1+[1+C\gamma]2\pi^2 Y|\ln Y| \bigg) +\Oh(\rho^2) \quad (n=4) \\
e_0(\rho)&\leq s_n a^{n-2}\rho +\Oh(\rho^2) \quad (n\geq 5),
\end{align*}
where
\begin{equation} \label{definitiona_gamma}
\gamma:=\int_{\R^n}V(x)|x|^{2-n}\, dx
\end{equation}
and $C>0$ is independent of $V$.
\end{theorem}

\noindent The assumptions on $V$ in Theorem \ref{thm:ESYndim} are presumably not optimal. In the actual grand-canonical calculation below, we only need $V$ and its Fourier transform to decay sufficiently fast at infinity (depending on the dimension), and of course the latter can be met by imposing finite smoothness on $V$. We use compact support of $V$ and $V(0)>0$ in Lemma \ref{lem-ensembles} below, which allows us to relate the canonical- and 'grand canonical' ground state energies. Presumably the assumption of compact support can be relaxed to a sufficiently fast decay.

In order to prove Theorem \ref{thm:ESYndim} we initially consider \eqref{model} with \emph{Dirichlet} boundary conditions. Our calculation below is carried out in the grand canonical ensemble, and hence we consider the second quantization of $H_{N,L}$
\begin{equation} \label{sdsdsdsmdnsdhsdjsdusdysdu}
H_L:=\bigoplus_{N=0}^{\infty}H_{N,L} \quad \text{on} \quad \F_L:=\bigoplus_{N=0}^{\infty} L^2_{\text{sym}}(\Lambda_L^N) ,
\end{equation}
with the corresponding 'grand canonical ground state energy'
\begin{equation} \label{ssdmsdjsdhsdn}
E_0^{\text{GC}}(N,L):=\inf\big\{\bra H_L \ket_{\Psi}: \|\Psi\|_{\F}=1, \, \bra\Num \ket_{\Psi} \geq N\big\},
\end{equation}
where $\Num=\Num_L$ denotes the number operator on $\F_L$ and $\bra A\ket_{\Psi}$ denotes the expectation $\bra \Psi,A\Psi\ket$ of any operator $A$ with $\Psi$ in its domain. Consider the canonical and grand canonical ground state energy \emph{per volume},
\begin{equation} \label{def:gsepervolume}
e_L(\rho):=\frac{E_0(\rho L^n,L)}{L^n}, \quad e_L^{\text{GC}}(\rho):=\frac{E_0^{\text{GC}}(\rho L^n,L)}{L^n}.
\end{equation}
We will assume that the limit
\begin{equation} \label{sdsdsdnsdbsdgsdhsfsdtsdy}
e(\rho):=\lim_{L\to \infty}e_L(\rho)
\end{equation}
is a convex function of $\rho$ (see e.g. \cite{ruelle}). The following result, which we prove in appendix \ref{app-ensembles}, shows that, in the thermodynamic limit, the canonical and grand canonical energies agree.

\begin{lemma} \label{lem-ensembles}
Suppose that $V\in L^1(\R^n)$ is nonnegative, radially symmetric and compactly supported. Suppose furthermore that $V\geq \epsilon \mychi_{B(0,R)}$, for some $\epsilon,R>0$. Then
$$
e(\rho)=\lim_{L\to \infty} e_L^{\text{GC}}(\rho).
$$
\end{lemma}

\noindent By \eqref{ssdmsdjsdhsdn} it is clear that $\rho\mapsto e_L^{\text{GC}}(\rho)$ is increasing, for any fixed $L$. As a consequence we have the following slightly stronger result.
\begin{corollary} \label{cor:sdsdsdmsdnsdh}
Suppose that $V$ satisfies the assumptions of Lemma \ref{lem-ensembles}, and suppose that $\rho_L\to \rho$ as $L\to \infty$. Then
$$
e(\rho)=\lim_{L\to \infty} e_L^{\text{GC}}(\rho_L)
$$
\end{corollary}

\begin{proof}
Fix an arbitrary $\epsilon>0$. By assumption $e(\rho)$ is convex and hence continuous. Thus we can choose $\delta>0$ such that 
$$
|e(\rho)-e(\rho')|\leq \epsilon,
$$
for each $\rho'>0$ with $|\rho-\rho'|\leq \delta$. Then, for $L$ sufficiently large,
\begin{eqnarray*}
e_L^{\text{GC}}(\rho_L)&\geq& e_L^{\text{GC}}(\rho-\delta) \\
&=& \big[e_L^{\text{GC}}(\rho-\delta)-e(\rho-\delta)\big] +e(\rho-\delta) \\
&\geq& \big[e_L^{\text{GC}}(\rho-\delta)-e(\rho-\delta)\big] +e(\rho)-\epsilon .
\end{eqnarray*}
By Lemma \ref{lem-ensembles} it then follows that
$$
\liminf_{L\to \infty} e_L^{\text{GC}}(\rho_L)\geq e(\rho)-\epsilon .
$$
Similarly we can show a consistent upper bound, and since $\epsilon$ was arbitrary, the result follows.
\end{proof}
In Section \ref{sec:trial} we construct a \emph{periodic} trial state with an expected number of particles $\bra \Num\ket =\rho L^n$, not directly leading to an upper bound on $e_0(\rho)$ via Lemma \ref{lem-ensembles}. However, Lemma \ref{eeruuriryruri} below, which is essentially proved in \cite{YY}, shows that given any periodic state, we can find a Dirichlet state on a slightly larger box, with almost as low energy. We let
$$
V_L(x):=\sum_{m\in \Z^n}V(x+mL)=\frac{1}{L^n}\sum_{p\in \Lambda_L^*}\hat V_pe^{ip\cdot x}, \quad x\in \R^n
$$
denote the $L$-periodization of $V$, where $\Lambda_L^*:=(2\pi/L)\Z^n$ and
$$
\hat V_p:=\int_{\R^n}e^{-ip\cdot x}V(x)\, dx
$$
denotes the Fourier transform of $V$, which is real-valued and radially symmetric, since $V$ is. Then let
$$
\tilde H_{N,L}:=\sum_{i=1}^N-\Delta_i+\sum_{1\leq j<k\leq N} V_L(x_j-x_k)
$$
with \emph{periodic} boundary conditions, and let $\tilde H_L$ denote its second quantization. Note that, since $V$ is nonnegative, it is clear that $V\leq V_L$, and hence the transition from $V$ to $V_L$ cannot decrease the energy. However, since $V_L\to V$ pointwise as $L\to \infty$, we expect the ground state energy of the two systems to coincide in the thermodynamic limit. 
\begin{lemma} \label{eeruuriryruri}
Let $L>2l>0$. Then
$$
E_0^{\text{GC}}(N,L+2l)\leq \bra \tilde H_{L} \ket_{\Psi} + C \frac{N}{lL},
$$
for each periodic, normalized $\Psi\in \F_L$ with $\bra \Num\ket_{\Psi}= N$. Here $C>0$ depends only on $n$.
\end{lemma} \label{periodic-dirichlet}

\noindent We apply Lemma \ref{eeruuriryruri} with $l:=\sqrt{L}/2$ and notice that
$$
\frac{E_0^{\text{GC}}(\rho L^n, L+2l)}{\rho L^n} = \frac{E_0^{\text{GC}}(\rho_{L+2l} (L+2l)^n, L+2l)}{\rho_{L+2l} (L+2l)^n} ,
$$
where 
$$
\rho_L:=\frac{\rho(L-2l)^n}{L^n} \to \rho \quad \text{as $L\to \infty$}.
$$
Together with Corollary \ref{cor:sdsdsdmsdnsdh} we conclude that
$$
e_0(\rho)\leq \limsup_{L\to \infty} \frac{\bra \tilde H_{L} \ket_{\Psi}}{\rho L^n},
$$
for each periodic, normalized $\Psi\in \F_L$ with expected number of particles $\bra \Num \ket_{\Psi}= \rho L^n$.

Finally, we note that, with the periodic potential $V_L$, we have (in the sense of quadratic forms)
\begin{equation}\label{second_q}
\tilde H_L=\sum_p p^2a_p^+a_p + \frac{1}{2L^n}\sum_{\stackrel{\text{\footnotesize{$p,q,r,s$}}}{\text{\footnotesize{$p+q=r+s$}}}}\hat V_{p-r}a_p^+a_q^+a_{r}a_{s},
\end{equation}
where all sums are over $\Lambda_L^*$ and where $a_p^+$ and $a_p$ denote the bosonic creation and annihilation operators on $\F_L$ w.r.t. the plane wave $x\mapsto L^{-n/2}e^{ip\cdot x}$.

\subsection[The Trial State]{The Trial State} \label{sec:trial}

The state in \eqref{trial_esy} can be defined as follows. Fix $\rho,L>0$ and set $N:=\rho |\Lambda|=\rho L^n$. Then let
\begin{equation} \label{sdsdsdmsdnsdjsdksdl}
\Psi :=\sum_{\alpha}f(\alpha)|\alpha \ket
\end{equation}
where  $\{|\alpha\ket \}_{\alpha }\subset \F$ is the orthonormal basis given by
$$
|\alpha\ket: =\prod_{k\in \Lambda^*}\frac{1 }{\sqrt{\alpha(k)!}} (a_k^+)^{\alpha(k)}  |0\ket,
$$
for each $\alpha:\Lambda^*\to \N_0$ with $|\alpha|:=\sum_{k\in \Lambda^*}\alpha(k)<\infty$. Note that, by the canonical commutation relations,
\begin{equation} \label{sddslsdspdolwnn}
a_p|\alpha \ket = \sqrt{\alpha(p)}|\alpha -\delta_p\ket \quad \text{and} \quad a_p^+|\alpha \ket = \sqrt{\alpha(p)+1}|\alpha +\delta_p\ket ,
\end{equation}
for any $p\in \Lambda^*$, where $\delta_p(k):=\delta_{p,k}$. Let
$$
\M:=\big\{\alpha: \Lambda^*\to \N_0 : |\alpha|<\infty \; \text{and $\alpha(-p)=\alpha(p)$ for each $p\in \Lambda^*$} \big\},
$$
We define the coefficient function $f$ in \eqref{sdsdsdmsdnsdjsdksdl} by 
\begin{equation} \label{sdsdksdlsdmsdj}
f(\alpha):= \exp\bigg(N_0+\sum_{p\neq 0}|\ln(1-c_p^2)| \bigg)^{-1/2}\cdot \bigg(\frac{N_0^{\alpha(0)}}{\alpha(0)!} \prod_{p\neq 0} c_p^{\alpha(p)}\bigg)^{1/2},
\end{equation}
for $\alpha\in \M$ and $f(\alpha)=0$ otherwise. Here $c:\Lambda^*\backslash \{0\}\to (-1,1)$ is to be chosen and 
\begin{equation} \label{sdsdmsdnsdj}
N_0:=N-\sum_{p\neq 0}\frac{c_p^2}{1-c_p^2}.
\end{equation}
It will be apparent later on that \eqref{sdsdmsdnsdj} is equivalent to the condition $\bra \Psi, \Num\Psi\ket=N$. We will assume that $c_{-p}=c_p$, for each $p$ and clearly we also need some decay of $c_p$ in order for the sums in \eqref{sdsdksdlsdmsdj} and \eqref{sdsdmsdnsdj} to converge. Given any operator $A$ with a domain containing $\Psi$, we let $\bra A\ket:=\bra \Psi, A\Psi\ket$ denote the expectation of $A$ in the state $\Psi$. Most of the interaction terms in \eqref{second_q} have zero expectation in the state $\Psi$. In fact, since $f$ vanishes outside $\M$ and since $\alpha(-p)=\alpha(p)$, for each $p\in \Lambda^*$ and each $\alpha\in \M$, it follows that only \emph{pair} interactions terms where either $p=r$, $p=s$ or $p=-q$ have nonzero expectation in $\Psi$. Thus
$$
\bra \tilde H_L\ket= \sum_{p}p^2\bra a_p^+a_p\ket+E_1+E_2+E_3,
$$
where
$$
E_1:= \frac{\hat V_0}{2|\Lambda|}\sum_{p,q}\bra a_p^+a_q^+a_pa_q\ket, \quad E_2:= \frac{1}{2|\Lambda|}\sum_{p\neq q}\hat V_{p-q}\bra a_p^+a_q^+a_pa_q\ket
$$
and
$$
E_3:= \frac{1}{2|\Lambda|}\sum_{p\neq \pm q}\hat V_{p-q}\bra a_p^+a_{-p}^+a_qa_{-q}\ket .
$$
Lemma \ref{lem:prop-trial} below provides us with all the relevant expectations in terms of $N_0$ and $c_p$. We introduce the notation
$$
h_p:=\frac{c_p^2}{1-c_p^2}  \quad \text{and} \quad s_p:= \frac{c_p}{1-c_p^2}.
$$

\begin{lemma}  \label{lem:prop-trial}
Let $p,q\in \Lambda^*$ with $p\neq \pm q$ and $p\neq 0$. Then
\begin{enumerate}
\item $\bra a_0^+a_0\ket =N_0=\bra a_0a_0\ket$ and $\bra a_0^+a_0a_0^+a_0\ket =N_0(N_0+1)$
\item $\bra a_p^+a_pa_q^+a_q\ket = \bra a_p^+a_p\ket \cdot \bra a_q^+a_q\ket$ 
\item $\bra a_p^+a_{-p}^+a_qa_{-q}\ket = \bra a_p^+a_{-p}^+\ket \cdot \bra a_qa_{-q}\ket$
\item $\bra a_p^+a_p\ket =h_p$
\item $ \bra a_p^+a_{-p}^+\ket=s_p$
\item $\bra a_p^+a_pa_{\pm p}^+a_{\pm p}\ket = h_p(2h_p+1)$
\end{enumerate}
\end{lemma}

\begin{proof}
The identities are proved similarly, so we only show a few of them. By definition of $\Psi$ and the relations \eqref{sddslsdspdolwnn}, we have
$$
\bra a_p^+a_p\ket =\sum_{\alpha}\alpha(p)|f(\alpha)|^2,
$$
for any $p\in \Lambda^*$. Define the operation $\A^0\alpha:=\alpha+\delta_0$ and $\A^p\alpha:=\alpha+\delta_p+\delta_{-p}$, for $p\neq 0$. Notice that
$$
f(\A^0\alpha)=N_0^{1/2}(\alpha(0)+1)^{-1/2}f(\alpha) \quad \text{and} \quad f(\A^p\alpha)=c_pf(\alpha).
$$
We then have
$$
\bra a_0^+a_0\ket =\sum_{\alpha\in \A^0(\M)}\alpha(0)|f(\alpha)|^2= \sum_{\beta}(\beta(0)+1)|f(\A^0\beta)|^2=N_0,
$$
where we have also used that $\sum_{\beta}|f(\beta)|^2=1$ due to normalization. For $p\neq 0$ we get 
$$
\bra a_p^+a_p\ket= \sum_{\beta}(\beta(p)+1)|f(\A^p\beta)|^2 =c_p^2(\bra a_p^+a_p\ket+1),
$$
and solving for $\bra a_p^+a_p\ket$ yields 4. Also,
$$
\bra a_p^+a_{-p}^+\ket =\sum_{\alpha}\overline{f(\A^p\alpha)}f(\alpha)(\alpha(p)+1) =c_p(h_p+1)=s_p,
$$
as claimed.
\end{proof}

Notice that, by Lemma \ref{lem:prop-trial},
$$
\bra \Num \ket= \sum_p\bra a_p^+a_p\ket = N_0+\sum_{p\neq 0}h_p,
$$
and hence the condition $\bra \Num \ket =N$  is indeed equivalent to \eqref{sdsdmsdnsdj}.

\subsection{Computation of the Energy}

Eventually we will choose $c_p$ via the new variable
$$
e_p:=\frac{c_p}{1+c_p}, \quad h_p=\frac{e_p^2}{1-2e_p}, \quad s_p=\frac{e_p(1-e_p)}{1-2e_p}.
$$
Note that the constraint $|c_p|<1$ is equivalent to $e_p<1/2$. In Lemma \ref{sdsdksdmsdh} below we calculate the energy $\bra \tilde H_L\ket$ per particle in the thermodynamic limit
$$
E(\rho):=\lim_{L\to \infty}\frac{\bra \tilde H_L\ket}{\rho L^n} .
$$
For this reason, it is convenient to assume that $e_p$ is independent of $L$, i.e. we assume that $c$ is defined on $\R^n\backslash\{0\}$ rather than on $\Lambda^*\backslash\{0\}$. We will also employ the fact that for any continuous function $F\in L^1(\R^n)$, decaying faster than $|p|^{-n-\epsilon}$ at infinity, for some $\epsilon>0$, we have the convergence
\begin{equation} \label{sdsdhsdjsdkdsl}
\lim_{L\to \infty}\frac{1}{L^n}\sum_{p\in \Lambda^*}F(p) =\frac{1}{(2\pi)^n}\int_{\R^n}F(p)\, dp .
\end{equation}
We denote the scattering solution by $1-w$ and set
$$
\phi :=Vw \quad \text{and} \quad g:=V-\phi=V(1-w).
$$
Note that $\hat g_0=2s_na^{n-2}$ by \eqref{sdsdsmsdnsdydudtdy}. Though $w$ is not integrable, it follows from the scattering equation \eqref{sdsdsdnsdbsdhsdg} that, as tempered distribution, $\hat w$ equals the function $p\mapsto \hat g_p/(2p^2)$. We shall abuse notation slightly by denoting
$$
\hat w_p:=\frac{\hat g_p}{2p^2}.
$$

\begin{lemma} \label{sdsdksdmsdh}
Suppose that $e:\R^n\backslash\{0\}\to (-\infty,1/2)$ is even, continuous and integrable with fast decay. Then
$$
E(\rho)=s_na^{n-2}\rho + Q+\tilde Q+\Omega,
$$
where
$$
Q:=\frac{1}{2(2\pi)^n\rho}\int p^2 \bigg[\frac{e_p^2+2\rho\hat w_pe_p}{1-2e_p}+(\rho \hat w_p)^2 \bigg]\, dp,
$$
$$
\tilde Q:=\frac{2}{(2\pi)^n}\int \hat \phi_p h_p\, dp,
$$
$$
\Omega:=\frac{1}{2(2\pi)^{2n}\rho}\int\int \big[ \hat V_{p-q}(e_p+\rho\hat w_p)(e_q+\rho\hat w_q)+ 2(\hat V_{p-q}-\hat V_p)s_ph_q -2\hat V_p h_ph_q\big]\, dp\, dq .
$$
\end{lemma}

\begin{proof}
By Lemma \ref{lem:prop-trial}, the kinetic energy is simply
$$
\sum_{p}p^2\bra a_p^+a_p\ket = \sum_{p\neq 0}p^2h_p = \sum_{p\neq 0}\frac{p^2e_p^2}{1-2e_p}.
$$
Using commutation relations, Lemma \ref{lem:prop-trial} and \eqref{sdsdmsdnsdj}, we find that
$$
E_1=\frac{\hat V_0}{2|\Lambda|}\bigg(\sum_{p,q}\bra a_p^+a_pa_q^+a_q\ket -\sum_p\bra a_p^+a_p \ket \bigg)= \frac{\hat V_0}{2|\Lambda|}\bigg(N^2+\sum_{p\neq 0}h_p(h_p+1) \bigg),
$$
where the last sum comes from the special cases $p=\pm q$. Note that contributions like that will vanish in the energy per particle in the thermodynamic limit. Similarly,
\begin{eqnarray*}
E_2&=&\frac{N_0}{|\Lambda|}\sum_{p\neq 0} \hat V_ph_p + \frac{1}{2|\Lambda|}\sum_{p\neq 0}\hat V_{2p}h_p(2h_p+1)+    \frac{1}{2|\Lambda|}\sum_{\stackrel{p,q\neq 0}{p\neq \pm q}}\hat V_{p-q} h_ph_q \\
&=& \sum_{p\neq 0}\rho \hat V_ph_p -\frac{1}{|\Lambda|}\sum_{p,q\neq 0}\hat V_ph_ph_q +  \frac{1}{2|\Lambda|}\sum_{p\neq 0}\hat V_{2p}h_p(2h_p+1)+ \frac{1}{2|\Lambda|}\sum_{\stackrel{p,q\neq 0}{p\neq \pm q}}\hat V_{p-q} h_ph_q ,
\end{eqnarray*}
and also
$$
E_3= \sum_{p\neq 0}\rho \hat V_ps_p -\frac{1}{|\Lambda|}\sum_{p,q\neq 0}\hat V_ps_ph_q + \frac{1}{2|\Lambda|}\sum_{\stackrel{p,q\neq 0}{p\neq \pm q}}\hat V_{p-q} s_ps_q .
$$
Thus, in the limit $L\to \infty$,
\begin{eqnarray*}
E(\rho)&=&\frac{\hat V_0 \rho}{2}+\frac{1}{(2\pi)^n\rho}\int\frac{p^2e_p^2+\rho\hat V_pe_p}{1-2e_p}\, dp +\frac{1}{2(2\pi)^{2n}\rho}\int\int \hat V_{p-q}s_ps_q-2\hat V_ps_ph_q \, dp\,dq \\
&+& \frac{1}{2(2\pi)^{2n}\rho}\int\int \hat V_{p-q}h_ph_q-2\hat V_ph_ph_q \, dp\,dq .
\end{eqnarray*}
By the relation $e_p=s_p-h_p$ we have
$$
\int\int \hat V_{p-q}s_ps_q-2\hat V_ps_ph_q \, dp\,dq =\int \int \hat V_{p-q}(e_pe_q -h_ph_q)+2(\hat V_{p-q}-\hat V_p)s_ph_q \, dp\, dq
$$
and hence
\begin{eqnarray*}
E(\rho) &=&\frac{\hat V_0 \rho}{2}+\frac{1}{(2\pi)^n\rho}\int\frac{p^2e_p^2+\rho\hat V_pe_p}{1-2e_p}\, dp +\frac{1}{2(2\pi)^{2n}\rho}\int\int \hat V_{p-q}e_pe_q\, dp\,dq \\
&+& \frac{1}{(2\pi)^{2n}\rho}\int\int  (\hat V_{p-q}-\hat V_p)s_ph_q -\hat V_ph_ph_q \, dp\,dq .
\end{eqnarray*}
Now, using $(2\pi)^n\hat \phi=\hat V*\hat w$, \eqref{sdsdsmsdnsdydudtdy} and $V=g+\phi$, we get
$$
\frac{\hat V_0\rho}{2}= \frac{\hat g_0\rho}{2}+\frac{\rho}{2(2\pi)^n}\int \hat V_p \hat w_p\, dp = s_na^{n-2}\rho+\frac{\rho}{2(2\pi)^n}\int \hat g_p \hat w_p\, dp + \frac{\rho}{2(2\pi)^n}\int \hat \phi_p \hat w_p\, dp
$$
and also
\begin{eqnarray*}
\frac{1}{2(2\pi)^{2n}\rho}\int\int \hat V_{p-q}e_pe_q\, dp\,dq &=& \frac{1}{2(2\pi)^{2n}\rho}\int\int \hat V_{p-q}(e_p+\rho \hat w_p)(e_q+\rho \hat w_q)\, dp\, dq \\
&-& \frac{1}{(2\pi)^n}\int \hat \phi_p e_p\, dp - \frac{\rho}{2(2\pi)^n}\int \hat \phi_p \hat w_p\, dp .
\end{eqnarray*}
Combining terms yields the desired.
\end{proof}

In \cite{esy08} the function $e_p$ is chosen as the pointwise minimizer of the sum of integrands in $Q$ and $\tilde Q$. However, it turns out that including the latter in the minimization problem does not lower the energy significantly. In fact, the calculation of Yau-Yin \cite{YY} suggests that $\tilde Q$ is really not present in the ground state energy, but should rather be cancelled by a term 'missing' in the energy of our trial state. Thus we will choose $e_p$ to minimize the simpler expression
$$
m_p:=\frac{e_p^2+2\rho \hat w_p e_p}{1-2e_p} .
$$
This yields
\begin{equation} \label{sdsdmsdksdlsdjsdk}
-e_p^2+e_p+\rho \hat w_p=0, \quad e_p= \frac{1}{2}\left(1-\sqrt{1+4\rho \hat w_p}\, \right) 
\end{equation}
and
$$
m_p=\frac{(1-2e_p)(-e_p-\rho\hat w_p)+(-e_p^2+e_p+\rho\hat w_p)}{1-2e_p}=\frac{1}{2}\big(\sqrt{1+4\rho\hat w_p}-1-2\rho\hat w_p \big),
$$
provided $1+4\rho \hat w_p\geq 0$. Note however that, since $\hat g$ is continuous, $\hat g_0>0$ and $\hat g_p\to 0$ as $|p|\to \infty$, it follows that $\hat w_p$ is bounded from below, and hence
$$
\liminf_{\rho\to 0}\big[ \inf_{p\neq 0}(1+4\rho \hat w_p)\big]\geq 1.
$$
With the choice in \eqref{sdsdmsdksdlsdjsdk} we have
\begin{equation} \label{sdsdksdjsdu}
Q=\frac{1}{2(2\pi)^n\rho}\int p^2\Phi(\rho\hat w_p)\, dp,
\end{equation}
where
\begin{equation} \label{sdsdmdsnsdjsdh}
\Phi(t):=\sqrt{1+4t}+2t^2-2t-1.
\end{equation}
Finally, we note that $|e_p|\leq \rho |\hat w_p|$, for $|p|\gg \rho^{1/2}$, and hence $e$ inherits decay from $\hat g$.

\subsection{Estimates}

In this section we estimate the integrals from Lemma \ref{sdsdksdmsdh} in the limit $\rho \to 0$, given the particular choice in \eqref{sdsdmsdksdlsdjsdk}. We begin with the term $Q$, and in fact we will derive asymptotics of order up to $\sim n/2$ with coefficients, all except one, given in terms of integrals of $\hat w_p$ (see also Table 1 below). We stress, however, that the physical relevance of these higher order asymptotics remains to be understood. In fact, while the main contribution in dimension three and four comes from $Q$, we believe that $\Omega$ and $Q$ are of the same leading order in dimension $n\geq 5$. 

\begin{lemma} \label{sdsdsdmsdksdisdu}

In dimension $n=3$,
$$
Q= (4\pi a \rho)\cdot \frac{128}{15\sqrt{\pi}} Y^{1/2} +o(\rho^{3/2}) \quad \text{and} \quad Q\leq (4\pi a \rho)\cdot \frac{128}{15\sqrt{\pi}} Y^{1/2}  .
$$
In dimension $n\geq 4$,
$$
Q= \sum_{m=3}^{\lceil n/2\rceil}c_m\rho^{m-1}+c_{n/2+1}\bigg( s_na^{n-2}\rho Y^{n/2-1}|\ln Y|\bigg) +\Oh(\rho^{n/2}),  
$$
where the error term depends on $V$ and where $c_m=0$ if $m\notin \N$, 
$$
c_m:=\frac{\Phi^{(m)}(0)}{2(2\pi)^nm!}\int_{\R^n}p^2\hat w_p^m, \quad m\leq (n+1)/2,
$$
and
$$
c_{n/2+1}:=\frac{\Phi^{(n/2+1)}(0)|\mathbb{S}^{n-1}|^{n/2+1}(n-2)^{n/2}}{4(2\pi)^n (n/2+1)!}.
$$
The function $\Phi$ is given in \eqref{sdsdmdsnsdjsdh}.
\end{lemma}

\begin{proof}
We first consider the case $n=3$. Let $\epsilon=(\hat g_0 \rho)^{1/2}$. By a change of variables $p\mapsto \epsilon p$, continuity of $\hat g$ and the dominated convergence theorem we have
\begin{equation} \label{sdsdsdmsdjsdi}
\rho^{-3/2}Q= \frac{\hat g_0^{5/2}}{2(2\pi)^3}\int_{\R^3}p^2\Phi\bigg(\frac{\hat g_{\epsilon p}}{2p^2\hat g_0}\bigg)\, dp \to \frac{\hat g_0^{5/2}}{2(2\pi)^3}\int_{\R^3}p^2\Phi\bigg(\frac{1}{2p^2}\bigg)\, dp
\end{equation}
as $\rho\to 0$. A direct calculation then yields
\begin{equation} \label{asasasnasbashg}
Q= (4\pi a \rho)\cdot \frac{128}{15\sqrt{\pi}} Y^{1/2} +o(\rho^{3/2}).
\end{equation}
The explicit upper bound claimed in the lemma can be obtained from the same calculation with the additional information that $\Phi$ is increasing and $\hat g_p\leq \hat g_0$. In \cite{esy08} the estimate is done more carefully and shows that \eqref{asasasnasbashg} holds with $o(\rho^{3/2})$ replaced by $\Oh(\rho^2|\ln \rho|)$. In higher dimensions the asymptotics of $Q$ is more subtle, due to the fact that the latter integral in \eqref{sdsdsdmsdjsdi} becomes divergent! That is, we cannot replace $\hat g_p$ by $\hat g_0$, because we need the decay of $\hat g$ in order for the integral to converge. However, from the asymptotics of $\Phi$ we get some information. First notice that $\Phi(t)=\Oh(t^2)$. Hence
\begin{eqnarray*}
Q_{\epsilon}'&:=&\frac{1}{2(2\pi)^n\rho}\int_{|p|\leq \epsilon}p^2\Phi(\rho\hat w_p)\, dp \leq \frac{1}{2(2\pi)^n\rho}\int_{|p|\leq \epsilon}p^22(\rho\hat w_p)^2\, dp \\
&\leq& \frac{\hat g_0^2\rho}{4(2\pi)^n} \int_{|p|\leq \epsilon}p^{-2}\, dp =Ca^{n-2}\rho Y^{n/2-1},
\end{eqnarray*}
where we have inserted $\hat g_0=2s_na^{n-2}$. To estimate $Q_{\epsilon}:=Q-Q_{\epsilon}'$, we expand $\Phi$ to the $(k-1)$ th order around $t=0$, where $k$ is the smallest integer such that $2k\geq n+3$. Since $\Phi(0)=\Phi'(0)=\Phi''(0)=0$, we have
$$
\Phi(t)= b_3t^3+\ldots +b_{k-1}t^{k-1}+\Oh(t^k),
$$
where $b_m:=\Phi^{(m)}(0)/m!$. Correspondingly we have the expansion
$$
Q_{\epsilon}=Q_{\epsilon}^{(3)}+\ldots +Q_{\epsilon}^{(k-1)}+ \E,
$$
where
$$
Q_{\epsilon}^{(m)}:=\frac{b_m}{2(2\pi)^n\rho}\int_{|p|>\epsilon}p^2(\rho \hat w_p)^m\, dp,
$$
and where
$$
|\E|\leq C\rho^{-1}\int_{|p|>\epsilon}|\rho \hat w_p|^k\, dp \leq C\hat g_0^k\rho^{k-1}\int_{|p|>\epsilon}p^{2-2k}\, dp= Ca^{n-2}\rho Y^{n/2-1}.
$$
If $m<n/2+1$, then $p^2\hat w_p^m$ is integrable at $p=0$, and we have
$$
Q_{\epsilon}^{(m)} = \frac{b_m\rho^{m-1}}{2(2\pi)^n}\int p^2\hat w_p^m\, dp +\Oh\big(a^{n-2}\rho Y^{n/2-1} \big).
$$
Notice that if $n$ is odd, then $k=(n+3)/2$, and hence $m<n/2+1$, for each $m\leq k-1$. In equal dimension there is a $m=n/2+1$ term:
\begin{eqnarray*}
Q_{\epsilon}^{(m)}&=& \frac{b_m}{2(2\pi)^n\rho}\int_{\epsilon<|p|\leq 1}p^2(\rho \hat w_p)^m\, dp + \Oh(\rho^{m-1}) \\
&=& \frac{b_m}{2(2\pi)^n\rho}\int_{\epsilon<|p|\leq 1}p^2 \bigg(\frac{\hat g_0}{2p^2} \bigg)^m  \, dp + \Oh(\rho^{m-1}),
\end{eqnarray*}
where the errors depend on $V$, and where we have used Lipschitz continuity of $\hat g$ to replace $\hat g_p$ with $\hat g_0$ in the second estimate. Now, by inserting $\hat g_0=2s_na^{n-2}$, we get
$$
Q_{\epsilon}^{n/2+1}= s_na^{n-2}\rho \frac{b_{n/2+1}|\mathbb{S}^{n-1}|^{n/2+1}(n-2)^{n/2}}{4(2\pi)^n}Y^{n/2-1}|\ln Y|+\Oh(\rho^{n/2}),
$$
where we have artificially replaced $|\ln(\hat g_0\rho)|$ by $|\ln Y|$ at the cost of an error of order $\rho^{n/2}$, depending on $V$. In particular, with $b_3=4$ and $|\mathbb{S}^{3}|=2\pi^2$, we have 
$$
Q_2^{(3)}=(4\pi^2a^2\rho)\cdot 2\pi^2Y |\ln Y| +\Oh(\rho^2),
$$
which is the term present in four dimensions.
\end{proof}

In table 1 we have listed the powers of $\rho$ in the expansion of $Q$ up to dimension $n=7$. Whether the expansion of $e_0(\rho)$ has this structure too remains to be clarified.

\begin{table}
\centering
\begin{tabular}{l |c ccc}
$n$ & $Q$  \\ \hline
$3$  & $\rho^{3/2}$ & $\Oh(\rho^2|\ln \rho|)$ \\
$4$ & $\rho^2|\ln \rho|$ & $\Oh(\rho^2)$ &  \\
$5$  & $\rho^2$ & $\Oh(\rho^{5/2})$  & \\
$6$  & $\rho^2$ & $\rho^3|\ln \rho|$ & $\Oh(\rho^3)$ &  \\
$7$ & $\rho^2$ & $\rho^3$ & $\Oh(\rho^{7/2})$ \\
$\vdots$ & $\vdots$
\end{tabular}
\caption{Qualitative expansion of $Q$ in the first few dimensions.}
\end{table}

\begin{lemma} \label{sdsdjdhdydhdy}
$$
\Omega(\rho)=\left\{\begin{array}{ll} \mathcal{O}(\rho^2|\ln \rho|)& n=3 \\  \mathcal{O}(\rho^2)& n\geq4   \end{array} \right.
$$
\end{lemma}

\begin{proof}
Using $|\hat V_p|\leq \hat V_0$, Lipschitz continuity of $\hat V$ and the relation in \eqref{sdsdmsdksdlsdjsdk}, we have
$$
|\Omega| \leq C_V\rho^{-1}\bigg\{ \bigg(\int e_p^2\, dp \bigg)^2+ \bigg(\int |s_p|\, dp \bigg)\bigg(\int h_q|q|\, dq \bigg)+\bigg(\int h_p\, dp \bigg)^2\bigg\} .
$$

Notice the asymptotics
$$
h_p=\frac{1}{2}\bigg(\frac{1+2\rho \hat w_p}{\sqrt{1+4\rho \hat w_p}} -1\bigg)= \left\{ \begin{array}{ll} \Oh\big(\sqrt{\rho \hat w_p}\big) & \text{as $|\rho\hat w_p|\to \infty$} \\ \Oh\big(\rho^2\hat w_p^2\big) & \text{as $|\rho \hat w_p|\to 0$}   \end{array}\right. .
$$
In fact, $h_p\leq C\rho^2\hat w_p^2$ for each $p\neq 0$, provided $\rho$ is sufficiently small. In dimension $n\geq 5$ we then simply have
$$
\int h_p\, dp =\Oh\big(\rho^2 \| \hat w \|_2^2 \big).
$$
Otherwise we split the integral into two parts 
$$
\int h_p\, dp \leq C\bigg[ \int_{|p|\leq \epsilon} (\rho \hat w_p)^{1/2}\, dp + \int_{|p|>\epsilon} (\rho \hat w_p)^2 \, dp \bigg]=I_1+I_2
$$
where $\epsilon:=(a\rho)^{1/2}$. Since $\hat g_p\leq \hat g_0=2s_na^{n-2}$, it follows that $I_1=\Oh(\rho Y^{n/2-1})$. In dimension $n=3$ we again use $|\hat g_p|\leq \hat g_0$ to obtain $I_2=\Oh(\rho Y^{1/2})$, and in four dimensions we get a logarithmic term, 
$$
I_2\leq C\rho Y|\ln Y| +C_V\rho^2.
$$
In total,
\begin{equation} \label{sdsdsdmdsjsdk}
\int h_p\, dp \leq \left\{ \begin{array}{ll} C \rho Y^{1/2}& n=3  \\ C \rho Y|\ln Y|+C_V\rho^2 & n=4 \\ C_V\rho^2 & n\geq 5  \end{array}\right. .
\end{equation}
By repeating the above estimates with an additional factor $|p|$ in the integrands, we see that
$$
\int |p| h_p\, dp \leq \left\{ \begin{array}{ll} C_V\rho^2|\ln \rho | & n=3 \\ C_V\rho^2 & n\geq 4  \end{array}\right. .
$$
The integral of $e_p^2$ is estimated similarly to $h_p$ and in fact \eqref{sdsdsdmdsjsdk} holds with $h_p$ replaced by $e_p^2$. Finally, since
$$
s_p=\frac{-\rho \hat w_p}{\sqrt{1+4\rho \hat w_p}},
$$
we see that 
$$
\int |s_p|\, dp =\mathcal{O}(\rho),
$$
for any $n\geq 3$, and we are done.
\end{proof}

\begin{remark}
From the estimate \eqref{sdsdsdmdsjsdk} and $|\hat \phi_p|\leq \hat \phi_0$ it follows that 
$$
\tilde Q(p) \leq \left\{ \begin{array}{ll} C \tilde \gamma a \rho Y^{1/2}& n=3  \\ C\tilde \gamma a^2 \rho Y|\ln Y|+C_V\rho^2 & n=4 \\ C_V\rho^2 & n\geq 5  \end{array}\right. ,
$$
where $\tilde \gamma:=\hat \phi_0/\hat g_0$.
\end{remark}

\noindent In order to finish the proof of Theorem \ref{thm:ESYndim} we only need to show that $\tilde \gamma \leq C \gamma$ as defined in  \eqref{definitiona_gamma}. This, however, follows easily from \eqref{sssdsbdgsdfsdtsdrsy} and \eqref{sdsdsmsdnsdydudtdy}, since then
$$
\hat \phi_0 =\int V(x)\phi(x)\, dx \leq a^{n-2}\int V(x)|x|^{2-n}\, dx = \frac{\hat g_0}{2s_n}\int V(x)|x|^{2-n}\, dx,
$$
as desired.

\appendix
\section{Equivalence of Ensembles} \label{app-ensembles}

In this section we prove Lemma \ref{lem-ensembles}. We will see that the canonical and grand canonical energies are related via the Legendre transform, and in order for this to be well-behaved globally, it is convenient to have high density bounds on the ground state energy. A trivial upper bound to $E_0(N,L)$ with periodic boundary conditions is obtained by calculating the energy of the constant function:
$$
E_0(N,L)\leq \frac{N(N-1)}{2|\Lambda|}\int V(x)\, dx .
$$
Thus, in the thermodynamic limit (and for all boundary conditions),
\begin{equation} \label{upperbound_constant}
e_0(\rho)\leq \frac{\hat V_0}{2}\rho .
\end{equation}
In the following lemma we derive a simple lower bound to $E_0(N,L)$ under the assumption that $V$ is uniformly strictly positive in a neighborhood of the origin. Due to lack of space, this forces a large fraction of the particles to interact.
\begin{lemma} \label{simple_lower}
Suppose that $V\geq \epsilon \mychi_{B(0,2R)}$, for some $\epsilon,R>0$. Then 
\begin{equation} \label{sdsdsdjsdnsdbsdvsdgsdf}
E_0(N,L)\geq C \epsilon R^n \frac{N^2}{|\Lambda|} - \frac{N}{2}V(0),
\end{equation}
for some constant $C>0$ depending only on the dimension.
\end{lemma}

\begin{proof}
We will simply discard the kinetic energy and show that the total interaction is pointwise bounded from below by the RHS in \eqref{sdsdsdjsdnsdbsdvsdgsdf}. Let $\mychi_R =\mychi_{B(0,R)}$. By Jensen's inequality we have
$$
\bigg(\int_{\Lambda}\sum_{j=1}^N \mychi_R(x_j-z)\, \frac{dz}{|\Lambda|} \bigg)^2\leq \frac{1}{|\Lambda|}\sum_{j,k}\int_{\Lambda}\mychi_R(x_j-z) \mychi_R(x_k-z)\, dz .
$$
However, the triangle inequality shows that
$$
\mychi_R(x_j-z)\mychi_R(x_k-z)\leq \mychi_{2R}(x_j-x_k)\mychi_R(x_k-z),
$$
and hence
\begin{eqnarray*}
\bigg(\int_{\Lambda}\sum_{j=1}^N \mychi_R(x_j-z)\, \frac{dz}{|\Lambda|} \bigg)^2&\leq& \frac{v_nR^n}{|\Lambda|} \sum_{j,k}\mychi_{2R}(x_j-x_k) \leq \frac{v_nR^n}{\epsilon |\Lambda|}\sum_{j,k}V(x_j-x_k) \\
&=& \frac{v_nR^n}{\epsilon |\Lambda|} \bigg(2\sum_{j<k}V(x_j-x_k)+NV(0) \bigg),
\end{eqnarray*}
where $v_n$ denotes the volume of the unit ball in $\R_n$. The result now follows by noting that
$$
\int_{\Lambda}\sum_{j=1}^N \mychi_R(x_j-z)\,dz \geq Nv_n 2^{-n}R^n,
$$
where the inequality and the factor $2^{-n}$ comes from the situation where $x_j$ is located close the corner of the box.
\end{proof}

Recall the notation in \eqref{def:gsepervolume} and \eqref{sdsdsdnsdbsdgsdhsfsdtsdy} for the ground state energy per volume. As a technical convenience we extend, for fixed $L>0$, the mapping $N\mapsto E_0(N,L)$ to $[0,\infty)$, as a piecewise linear function, by setting $E(0,L)=0$ and 
\begin{equation} \label{sdsdsdmsdnsdgsdhsdtsdysu}
E_0(N+\sigma,L)=(1-\sigma)E_0(N,L)+\sigma E_0(N+1,L), \quad \sigma\in [0,1].
\end{equation}
Note that, as a consequence of Lemma \ref{simple_lower} we have the lower bounds
\begin{equation} \label{sdsdsdmsdnsdhsdy}
e_L(\rho) \geq C_1\rho^2-C_2\rho \quad \text{and} \quad e(\rho) \geq C_1\rho^2-C_2\rho ,
\end{equation}
for constants $C_1,C_2>0$ depending on $V$. 

Since each $N$-particle sector is naturally imbedded in the Fock space, it follows that $E_0^{\text{GC}}(N,L)\leq E_0(N,L)$. We remark that in case $N$ is not a natural number, the inequality follows from the convention \eqref{sdsdsdmsdnsdgsdhsdtsdysu} by considering the combination
$$
\Psi := \sqrt{1-\sigma}\Psi_{\lfloor N\rfloor} +\sqrt{\sigma}\Psi_{\lceil N \rceil }
$$
of arbitrary $\lfloor N\rfloor$-particle and $\lceil N \rceil $-particles states, where $\sigma:= N-\lfloor N\rfloor$. In order to prove Lemma \ref{lem-ensembles}, we therefore only need to show that
\begin{equation} \label{sdsdskdslsdmsdnsdhsdy}
\liminf_{L\to \infty} e_L^{GC}(\rho) \geq e(\rho).
\end{equation}
We introduce a chemical potential $\mu\geq 0$ and notice that, for any normalized $\Psi=(\Psi_0,\Psi_1,\ldots)\in \F$ with $\bra \Psi, \Num\Psi \ket\geq \rho L^n$ we have the lower bound
\begin{eqnarray*}
\frac{\bra \Psi , H_L\Psi \ket }{L^n}&=& \frac{1}{L^n} \big[\mu \bra \Psi, \Num\Psi \ket + \bra \Psi, (H_L-\mu \Num)\Psi \ket   \big] \\
&\geq& \mu \rho +\sum_{N=0}^{\infty}\| \Psi_N\|^2 \bigg[ e_L(N/L^n)-\mu\frac{N}{L^n}\bigg] \\
&\geq& \mu \rho +f_L(\mu), 
\end{eqnarray*}
where $f_L:=- e_L^*$ and where
$$
 g^*(\mu):=\sup_{\rho \geq 0} \big[\mu \rho -g(\rho) \big],
$$
denotes the Legendre Transform of any function $g:[0,\infty)\to \R$, and for $\mu\geq 0$ such that the supremum is finite. We will employ the well-known fact \cite{simon_convex} that the Legendre transform is involute on convex functions, meaning that $( g^*)^*=g^*$. The inequality \eqref{sdsdskdslsdmsdnsdhsdy} will then follow, provided we can show the convergence
$$
\lim_{L\to \infty} f_L(\mu) =f(\mu):=- e^*(\mu),
$$
for each $\mu\geq 0$. Now, by definition,
$$
f_L(\mu)\leq e(\rho)-\mu \rho +[e_L(\rho)-e(\rho)],
$$
and hence
$$
\limsup_{L\to \infty} f_L(\mu)\leq e(\rho)-\mu \rho ,
$$
for each $\rho\geq 0$. It follows that
$$
\limsup_{L\to \infty} f_L(\mu)\leq f(\mu).
$$

For the lower bound we employ the following lemma.

\begin{lemma} \label{duplicate}
Suppose that $V$ is compactly supported with, say, $\text{supp}(V)\subset B(0,R)$. Then 
$$
e_L(\rho)\geq (1+R/L)^ne(\rho [1+R/L]^{-n}) 
$$
for each $\rho,L>0$. 
\end{lemma}

\begin{proof}
By convexity of $e(\rho)$ we may assume that $N:=\rho L^n$ is an integer. Let $k\in \N$ and put $L'=k(L+R)$.We can place $M:=k^n$ copies of the box $\Lambda_L$ inside the larger box $\Lambda_{L'}$ with separation $R$ between neighboring boxes. From an $N$-particle trial state $\Psi$ in $\Lambda_L$ we can construct a trial state with $MN$ particles by placing independent particles in each of the $M$ boxes, each with state $\Psi$. Because of the Dirichlet boundary condition, this gives a trial state on $\Lambda_{L'}$ by extending $\Psi$ by zero and, due to the separation, particles in different boxes do not interact. Minimizing over $\Psi$ yields
$$
e_L(\rho)\geq (1+R/L)^n e_{L'}(\rho[1+R/L]^{-n}).
$$
This estimate holds for each $k\in \N$, so the result follows by taking the limit $k\to \infty$.
\end{proof}

By Lemma \ref{duplicate} we have
 \begin{eqnarray*}
e_L(\rho)-\mu \rho &\geq& e(\rho_L)-\mu \rho \\
 &=& [1+R/L]^n(e(\rho_L)-\mu \rho_L) +\epsilon_L \\
 &\geq &  [1+R/L]^n f(\mu)+\epsilon_L,
 \end{eqnarray*}
where 
$$
\rho_L:= \rho [1+R/L]^{-n} \quad \text{and} \quad \epsilon_L:=e(\rho_L)(1-[1+R/L]^n) .
$$
Now notice that, by \eqref{sdsdsdmsdnsdhsdy},
$$
f_L(\mu)=\inf_{\rho\in [0,\rho_{\mu}]} [e_L(\rho)-\mu\rho],
$$
for some $\rho_{\mu}>0$. From the upper bound \eqref{upperbound_constant}, we then have
$$
f_L \geq [1+R/L]^n f(\mu) + C\rho_{\mu}^2(1-[1+R/L]^n), 
$$
and consequently
$$
\liminf_{L\to \infty} f_L(\mu) \geq f(\mu),
$$ 
as desired.

\section*{Acknowledgements}
A. Aaen was partially supported by the Lundbeck
 Foundation and the European 
Research Council under the
European Community's Seventh Framework Program (FP7/2007--2013)/ERC grant
agreement  202859.

\bibliographystyle{plain}
\bibliography{bibliography}

\begin{thebibliography}{10}

\bibitem{bec1}
M.H. Anderson, J.R. Ensher, M.R. Matthews, C.E. Wieman, and E.A. Cornell.
\newblock Observation of {B}ose–{E}instein condensation in a dilute atomic
  vapor.
\newblock {\em Science}, 269, 1995.

\bibitem{Bo}
N.N. Bogoliubov.
\newblock On the theory of superfluidity.
\newblock {\em J. Phys. (USSR)}, 11, 1947.

\bibitem{dyson}
F.~Dyson.
\newblock Ground-state energy of a hard-sphere gas.
\newblock {\em Phys Rev.}, 106:20--26, 1957.

\bibitem{esy08}
L.~Erd\Horig{o}s, B.~Schlein, and H-T. Yau.
\newblock The ground state energy of a low density {B}ose gas: a second order
  upper bound.
\newblock {\em Phys. Rev. A}, 78, 2008.

\bibitem{G-A}
M.~Girardeau and R.~Arnowitt.
\newblock Theory of many-boson systems: Pair theory.
\newblock {\em Phys. Rev. 105}, 113:755--760, 1959.

\bibitem{giulianiseiringer}
A.~Giuliani and R.~Seiringer.
\newblock The ground state energy of the weakly interacting {B}ose gas at high
  density.
\newblock {\em J. Stat. Phys.}, 135:915--934, 2009.

\bibitem{lee-yin}
J.~O. Lee and J.~Yin.
\newblock A lower bound on the ground state energy of a dilute {B}ose gas.
\newblock {\em J. Math. Phys.}, 51, 2010.

\bibitem{leehuangyang}
T.D. Lee, K.~Huang, and C.N. Yang.
\newblock Eigenvalues and eigenfunctions of a {B}ose system of hard spheres and
  its low-temperature properties.
\newblock {\em Phys. Rev.}, 106:1135--1145, 1957.

\bibitem{leeyang}
T.D. Lee and C.N. Yang.
\newblock Many-body problem in quantum mechanics and quantum statistical
  mechanics.
\newblock {\em Phys Rev.}, 105:119--120, 1957.

\bibitem{LSSY}
E.~Lieb, R.~Seiringer, J.P. Solovej, and J.~Yngvason.
\newblock {\em The Mathematics of the Bose Gas and its Condensation}.
\newblock Birkh\"{a}user, 1. edition, 2005.

\bibitem{lieblinger}
E.~H. Lieb and W.~Liniger.
\newblock Simplified approach to the ground-state energy of an imperfect {B}ose
  gas. application to the one-dimensional model.
\newblock {\em Phys. Rev.}, 134 2A, 1963.

\bibitem{liebseiringeryngvason2000}
E.~H. Lieb, R.~Seiringer, and J.~Yngvason.
\newblock Bosons in a trap: A rigorous derivation of the {G}ross-{P}itaevskii
  energy functional.
\newblock {\em Phys. Rev. A}, 61 043602, 2000.

\bibitem{liebyngvason98}
E.~H. Lieb and J.~Yngvason.
\newblock Ground state energy of the low density {B}ose gas.
\newblock {\em Phys. Rev. Lett.}, 80:2504--2507, 1998.

\bibitem{liebyngvason2001}
E.~H. Lieb and J.~Yngvason.
\newblock The ground state energy of a dilute two-dimensional {B}ose gas.
\newblock {\em J. Stat. Phys.}, 103:509--526, 2001.

\bibitem{mc}
C.~Mora and Y.~Castin.
\newblock Ground state energy of the two-dimensional weakly interacting {B}ose
  gas: First correction beyond bogoliubov theory.
\newblock {\em Phys. Rev. Lett.}, 102, 2009.

\bibitem{RS4}
M.Reed and B.~Simon.
\newblock {\em Methods of Modern Mathematical Physics Vol. 4}.
\newblock Academic Press, 1. edition, 1978.

\bibitem{ruelle}
D.~Ruelle.
\newblock {\em Statistical Mechanics: Rigorous Results}.
\newblock Imperial College Press and World Scientific, 3. edition, 1969.

\bibitem{schick}
M.~Schick.
\newblock Two-dimensional system of hard core bosons.
\newblock {\em Phys. Rev. A}, pages 1067--1073, 1971.

\bibitem{simon_convex}
B.~Simon.
\newblock {\em Convexity: An Analytical Viewpoint}.
\newblock Cambridge University Press, 1. edition, 2011.

\bibitem{Solovej}
J.P. Solovej.
\newblock Upper bounds to the ground state energies of the one- and two-
  component charged {B}ose gases.
\newblock {\em Comm. Math. Phys.}, 266, 2006.

\bibitem{yang}
C.~N. Yang.
\newblock Pseudopotential method and dilute hard "sphere" {B}ose gas in
  dimension 2,4 and 5.
\newblock {\em Europhys. Lett.}, 84 40001, 2008.

\bibitem{YY}
H-T. Yau and J.~Yin.
\newblock The second order upper bound for the ground energy of a {B}ose gas.
\newblock {\em J. Stat. Phys.}, 136 (3):453--503, 2009.

\bibitem{yinphd}
J.~Yin.
\newblock Quantum many-body systems with short-range interactions (ph.d.
  dissertation).
\newblock {\em Princeton University.}, 2008.

\end{thebibliography}

\end{document}